\documentclass[a4paper,USenglish,11pt,DIV10,final]{scrartcl}

\usepackage{stmaryrd}
\usepackage{csquotes}
\usepackage{colonequals}
\usepackage{amsmath}
\usepackage{amsthm}
\usepackage{amssymb}
\usepackage{thm-restate}
\usepackage{todonotes}
\usepackage{bm}
\usepackage{booktabs}
\usepackage{xspace}
\usepackage[final]{microtype}
\usepackage{xcolor}
\usepackage{multirow}
\usepackage{enumitem}
\setlist{itemsep=-2pt,topsep=2pt}

\theoremstyle{plain}
\newtheorem{theorem}{Theorem}[section]
\newtheorem{lemma}[theorem]{Lemma}

\theoremstyle{definition}
\newtheorem{definition}[theorem]{Definition}

\theoremstyle{remark}

\newcommand{\set}[2]{\left\{#1\mathrel{\left|\vphantom{#1}\vphantom{#2}\right.}#2\right\}}

\newcommand{\smallset}[1]{\left\{#1\right\}}

\newcommand*{\words}[1][\Lambda]{#1{}^{\mathsf{LO}}}
\newcommand*{\terms}[1][\Lambda]{T_{#1}}

\newcommand*{\enc}[1]{\langle #1\rangle}
\newcommand*{\Fragm}{\mathcal{F}}
\newcommand*{\FragmG}{\mathcal{G}}
\newcommand*{\FO}{\mathsf{FO}}

\newcommand*{\Varset}[1][X]{\mathbb{#1}}
\newcommand*{\Vars}{\Varset[V]}
\newcommand*{\Nat}{\mathbb{N}}
\newcommand*{\Rat}{\mathbb{Q}}
\newcommand*{\Int}{\mathbb{Z}}
\newcommand*{\Quantifiers}{\mathcal{Q}}
\newcommand*{\quant}{\mathsf{Q}}

\newcommand*{\dis}[3]{#3/#1 #2}
\newcommand*{\dom}[1]{\operatorname{dom}(#1)}
\newcommand*{\repl}[3]{#1[#2/#3]}
\newcommand*{\fv}{\operatorname{F\kern-0.25emV}}

\newcommand*{\qd}{\operatorname{qd}}
\newcommand*{\theLO}{\varrho}
\newcommand*{\twoVarLO}{{\sigma}}
\newcommand*{\twoVarLOText}{$\sigma$}

\newcommand*{\move}[5]{#1[#2#3,#4,#5]}

\newcommand*{\Dwins}[1]{\mathrel{\lesssim_{#1}}}
\newcommand*{\Swins}[1]{\mathrel{\not\lesssim_{#1}}}

\newcommand*{\evalt}[2]{\llbracket #1\rrbracket_{#2}}

\newcommand*{\ord}{\operatorname{ord}}

\newcommand*{\EF}{Ehrenfeucht-Fra\"iss\'e\xspace}

\newcommand*{\para}[1]{\medskip\par\noindent{\it\normalsize #1}}
\newcommand*{\lpara}[2]{\para{#1:} #2.\quad}
\newcommand*{\basecase}[1]{\lpara{Base case}{#1}}
\newcommand*{\inductivestep}[1]{\lpara{Inductive step}{#1}}

\newcounter{casecounter}
\newcounter{subcasecounter}
\newcommand*{\resetcases}{\setcounter{casecounter}{0}}
\newcommand*{\case}[1]{\stepcounter{casecounter}\setcounter{subcasecounter}{0}\lpara{Case \arabic{casecounter}}{#1}}

\newcommand*{\refenum}[1]{{#1.{}}}

\begin{document}

\title{Two-Variable Ehrenfeucht-Fra\"iss\'e Games over Omega-Terms}

\author{Manfred Kuf\-leitner \and Jan Philipp W\"achter}

\date{University of Stuttgart, FMI}

\maketitle

\begin{abstract}
  \noindent
  \textsf{\textbf{Abstract.}} \ 
   Fragments of first-order logic over words can often be characterized
  in terms of finite monoids, and identities of omega-terms are an effective mechanism for specifying classes of monoids. Huschenbett and the first author have shown how to use infinite \EF games on linear orders for showing that some given fragment satisfies an identity of omega-terms (STACS 2014). After revisiting this result, we show that for two-variable logic one can use simpler linear orders.
\end{abstract}

\section{Introduction}

As many fragments of first-order logic correspond to sets of finite monoids, the study of such
fragments often results in the study of algebraic properties of these monoids. A recurring problem
in this connection is the question whether a given identity of omega-terms $u = v$ holds. The 
omega-terms $u$ and $v$ in the identity are finite words over an alphabet of variables with an
additional $\omega$-power. This $\omega$-power is usually interpreted as a mapping of monoid
elements which maps the element $m$ of a finite monoid $M$ onto its generated idempotent 
$m^\omega$, i.\,e.~the element $e = m^k \in M$ for the smallest positive number $k$ such that $e =
e^2$, which is known to exist for any element in a finite monoid. With this interpretation the
identity $u = v$ holds in $M$ if $u$ and $v$ yield the same element for any assignment of the
variables occurring in $u$ and $v$ onto elements of $M$ and later interpretation of the 
concatenation with $M$'s binary operation.

To avoid confusion with the second historical meaning of $\omega$ to denote the linear order type
of all natural numbers, omega-terms will be called $\pi$-terms in this paper. This notation was introduced by Perrin and Pin~\cite{pp04:short} in the context of infinite words. Of course, the $\omega$-power will also be called $\pi$-power instead.

Often one can use \EF games in order to show that a certain identity holds in all monoids belonging
to a certain first-order logic fragment. This is usually done by describing a winning strategy for
Duplicator in the $n$-round \EF game for an arbitrary $n$. This approach usually involves a certain 
amount of book-keeping to keep track of the remaining rounds. M.~Huschenbett and the first author
presented a way to avoid this book-keeping by playing infinite long \EF games \cite{HuschenbettKufleitner14stacs:short}.
Here, the $\pi$-power gets replaced by a certain linear order type, namely $\theLO = \omega + 
\zeta \cdot \eta + \omega^*$, where $\omega$ denotes the linear order type of the natural number,
$\omega^*$ its dual, $\zeta$ that of all integers and $\eta$ that belonging to the rational numbers.
This leads to the notion of generalized words which are not necessarily finite anymore.

In this paper we are first going to revisit this result and then have a further look on how the
linear ordering can be simplified when we restrict our consideration onto fragments of first-order
logic with only two variables.

\section{Preliminaries}

\subparagraph*{Generalized Words.}
A central notion for this paper is that of generalized words. We therefore fix an alphabet
$\Lambda$. A 
\emph{generalized word} $w$ over $\Lambda$ is a triple $w = (P_w, \leq_w, l_w)$ where $P_w$ is
a set of \emph{positions}, $\leq_w$ is a linear order relation on $P_w$ (i.\,e.~a total, reflexive,
transitive and anti-symmetric binary relation), and $l_w: P_w \to \Lambda$ is a labeling map.
The set $P_w$ of positions of $w$ is called the \emph{domain} of $w$, denoted by $\dom{w} = P_w$.
To simplify notation, we write $w(p)$ instead of $l_w(p)$ for a position $p \in P_w$. The
\emph{order type} of $w$ is the isomorphism\footnote{Two ordered sets are \emph{isomorphic} if
there is an order-preserving bijection between them.} type of the linear ordered set $(P_w, 
\leq_w)$. Any finite word $u = a_1 a_2 \dots a_n$ of length $n \in \Nat$ in the set of finite words 
over $\Lambda$, denoted by $\Lambda^*$, can be regarded as a generalized word: its domain is the set
$\dom{u} = \smallset{1, 2, \dots, n}$, the order relation is the natural order, and the labeling
map labels position $i \in \dom{u}$ by $a_i$. We therefore see finite words as generalized words
in the following, and refer by \enquote{word} to generalized words instead of finite words.
We follow the custom of identifying the order type of any finite word of length $n \in \Nat$
(i.\,e.~the order isomorphism type of the naturally ordered set $\smallset{1, 2, \dots, n}$) with
$n$.

We consider two (generalized) words $u$ and $v$ \emph{isomorphic}, if there is an isomorphism 
$\iota: \dom{u} \to \dom{v}$ of linear ordered sets which is compatible with the labeling, 
i.\,e.~for which $u(p) = v(\iota(p))$ holds for all $p \in \dom{u}$. We do not distinguish between
isomorphic words anymore, and denote the set of (isomorphism classes of) \emph{countable} words 
by $\words$; the $\mathsf{LO}$ is for \emph{linear order}. As mentioned above, we consider 
$\Lambda^*$ as a subset of $\words$.

Given a linear ordered set $(T, \leq_T)$ of isomorphism type $\tau$ we can create the $\tau$-power
$w^\tau$ of any word $w \in \words$: the domain of $w^\tau$ is $\dom{w} \times T$, the linear
order relation $\leq_{w^\tau}$ is defined by $(p, t) \leq_{w^\tau} (p', t')$ if $t <_T t'$ 
(i\,e.~$t \leq_t t'$ and $t \neq t'$) or if $t = t'$ and $p \leq_w p'$, and the labeling is given 
by $w^\tau(p, t) = w(p)$. One should observe that this coincides with the usual $n$-power for
a natural number $n \in \Nat$ over finite words.

\subparagraph*{Logic over Words.}
Next, we give a brief overview of logic over words. We fix a countably infinite $\Vars = 
\smallset{x, y, z, \dotsc}$ of first-order variables. A first-order logical formula over words is
syntactically defined by
\begin{align*}
	\varphi &\coloncolonequals \top \mid \bot \mid x = y \mid \lambda(x) = a \mid
			x < y \mid x \leq y \mid 
	 \neg \varphi \mid \varphi \lor \varphi \mid \varphi \land \varphi \mid
			\exists x\, \varphi \mid \forall x\, \varphi \,,
\end{align*}
with $x, y \in \Vars$ and $a \in \Lambda$. The set of all these formulae is $\FO$. A \emph{sentence}
is a formula $\varphi \in \FO$ with $\fv(\varphi) = \emptyset$ where $\fv(\varphi)$ denotes the set
of free variables in $\varphi$ which are defined as usual.

For the semantics of first-order formulae we define an \emph{$\Varset$-valuation} for a finite 
subset $\Varset \subseteq \Vars$ on a word $w \in \words$ as the pair $\enc{w, \alpha}$ where 
$\alpha: \Varset \to \dom{w}$ is an arbitrary map. Whether the $\Varset$-valuation satisfies a 
formula $\varphi \in \FO$ with $\fv(\varphi) \subseteq \Varset$ depends on the following 
interpretations:
\begin{itemize}
  \item Variables get values in $\dom{w}$; free variables are interpreted according to $\alpha$.
  \item $\top$ is always satisfied, and $\bot$ is never satisfied.
  \item The logical operations $\neg$, $\lor$, and $\land$ are interpreted as usual, just
        like the quantifiers $\exists$ and $\forall$.
  \item The binary operators of comparison $=$, $<$ and $\leq$ are interpreted according to the
        ordering of $\dom{w}$ with respect to $w$'s order type.
  \item $\lambda$ gets interpreted by the labeling map $l_w$ of $w$.
\end{itemize}
If $\enc{w, \alpha}$ satisfies a formula $\varphi \in \FO$ it is called a \emph{model} and we 
write $\enc{u, \alpha} \models \varphi$. Any word $w \in \words$ is identified by $\enc{w, 
\emptyset}$, which is unique. This allows for the notation $w \models \varphi$ for a first-order
sentence $\varphi$.

\subparagraph*{Fragments.} 
The first author and A.~Lauser proposed a formal definition of logical fragments as sets of
formulae satisfying certain closure properties \cite{KufleitnerL12icalp:short}. For this paper
(just like in \cite{HuschenbettKufleitner14stacs:short}) we are going to use a slightly modified 
version of this idea. A first-order formula $\mu$ in which a special additional constant predicate
$\circ$ appears exactly once is called a \emph{context}. The predicate $\circ$ is intended as a
placeholder which can be replaced by other first-order formulae. The result of replacing $\circ$
in a context $\mu$ by a formula $\varphi \in \FO$ is denoted by $\mu(\varphi)$. The notion of a
context now allows for the formal definition of a fragment.

\begin{definition}
\label{def:fragment}
A \emph{fragment} is a non-empty set of formulae $\Fragm \subseteq
\FO$ such that for all contexts $\mu$, formulae $\varphi,\psi \in
\FO$, $a \in \Lambda$ and $x,y \in \Vars$ the following conditions are
satisfied:
\begin{enumerate}
\item If $\mu(\varphi) \in \Fragm$, then $\mu(\top) \in \Fragm$,
  $\mu(\bot) \in \Fragm$, and $\mu(\lambda(x) = a) \in \Fragm$.
\item $\mu(\varphi\lor\psi) \in \Fragm$ if, and only if, $\mu(\varphi)
  \in \Fragm$ and $\mu(\psi) \in \Fragm$.
\item $\mu(\varphi\land\psi) \in \Fragm$ if, and only if,
  $\mu(\varphi) \in \Fragm$ and $\mu(\psi) \in \Fragm$.
\item \label{cond:negation} If $\mu(\neg\neg\varphi) \in \Fragm$, then
  $\mu(\varphi) \in \Fragm$.
\item If $\mu(\exists x\,\varphi) \in \Fragm$ and $x \not\in
  \fv(\varphi)$, then $\mu(\varphi) \in \Fragm$.
\item \label{cond:forall} If $\mu(\forall x\,\varphi) \in \Fragm$ and
  $x \not\in \fv(\varphi)$, then $\mu(\varphi) \in \Fragm$.
\item $\mu(x < y) \in \Fragm$ if, and only if, $\mu(x \leq y) \in
  \Fragm$.
\end{enumerate}
\end{definition}

\noindent
These closure properties are satisfied by many classes of formulae which arise naturally in the
study of first-order logic, such as $\FO$ itself, the set of first-order formulae with limited 
quantifier alternations (e.\,g.~$\Sigma_n$ and $\Pi_n$), and the set of those whose number of 
variables is bounded by a natural number $m$ (denoted by $\FO^m$).

We define the \emph{quantifier depth} $\qd(\varphi)$ of a formula $\varphi \in \FO$ in the usual 
manner and say that a fragment $\Fragm$ has \emph{bounded quantifier depth} if there is an $n \in
\Nat$ such that $\qd(\varphi) \leq n$ for all $\varphi \in \Fragm$. For any fragment $\Fragm$ and
any natural number $n$ we define $\Fragm_n = \set{\varphi}{\varphi \in \Fragm, \qd(\varphi) \leq
n}$ as the corresponding fragment with bounded quantifier depth $n$. One should verify that
$\Fragm_n$ indeed is a fragment. An example of such a fragment is $\FO_n$, the set of first-order
formulae whose quantifier depth is bounded by $n$.

\section{Ehrenfeucht-Fra\"iss\'e Games}
\label{sec:EF_game}

In this section we will have a closer look on \EF games. These games can often be used to show 
whether a given fragment of first-order logic can distinguish two words, i.\,e.~whether there is
a formula in the fragment for which one word is a model while the other is not.

We define the set of quantifiers as $\Quantifiers = \{\exists,\forall,\neg\exists,\neg\forall\}$.
Note, that we also consider the \enquote{negated quantifiers} $\neq\exists$ and $\neq\forall$ as
quantifiers. Let $\Fragm$ be a fragment, $x \in \Vars$ and $\quant \in \Quantifiers$, then the
\emph{reduct} of $\Fragm$ by $\quant x$ is the set
$$
  \dis\quant x\Fragm = \set{ \varphi \in \FO }{ \quant x\,\varphi \in \Fragm } \,.
$$
Note that a reduct of a fragment is either empty or a fragment itself.

We call the \EF game which is played over the fragment $\Fragm$ and on the words $u, v \in \words$
the \emph{$\Fragm$-game on $(u, v)$}. Configurations $(\FragmG, \enc{u, \alpha}, \enc{v, \beta})$ 
of such a game consist of
\begin{enumerate}
  \item an iterated, non-empty reduct $\FragmG$ of $\Fragm$,
  \item a $\Varset$-valuation $\enc{u, \alpha}$, and
  \item a $\Varset$-valuation $\enc{v, \beta}$
\end{enumerate}
for a \emph{finite} subset $\Varset \subseteq \Vars$ of variables. Such a configuration is also
called an \emph{$\Varset$-configuration} of the game.

In the beginning the game is in the $\emptyset$-configuration $(\Fragm, u, v)$. The game can be
played any number -- even an infinite number -- of rounds. In each round a variable gets assigned
a new value; this variable can either be a variable which so far had no value or a variable which
already has one. In more detail: suppose the game is in the $\Varset-$configuration $S = (\FragmG, 
\enc{u, \alpha}, \enc{v, \beta})$. One round consists of two steps:
\begin{enumerate}
  \item The first player, called \emph{Spoiler}, chooses
        \begin{itemize}
          \item a quantifier $\quant \in \Quantifiers$,
          \item a variable $x \in \Vars$ such that $\dis\quant x\FragmG \not= \emptyset$, and
          \item a \emph{quest} $q$.
        \end{itemize}
        The quest is either a position in $\dom{u}$ if $\quant \in \smallset{ \exists, \neg\forall }$
        or a position in $\dom{v}$ if $\quant \in \smallset{ \forall, \neg\exists }$.
  \item The second player, called \emph{Duplicator}, replies with a position $r$ in the domain
        of the other word (with respect to the quest).
\end{enumerate}
The follow-up configuration $\move S\quant x q r$ of the game is defined in 
Table~\ref{tab:masterplan}.
\begin{table}[t]
  \centering
  \begin{tabular}{r@{{\;}={\;}}lccc} \toprule
    \multicolumn{3}{c}{\multirow{2}{*}{Spoiler}} & \multirow{2}{*}{Duplicator} & resulting configuration \\ 
    \multicolumn{2}{c}{}&&& $\move S\quant x q r$ \\ \midrule
$\quant x$ & $\exists x$     & $q \in \dom{u}$ & $r \in \dom{v}$ & $(\dis\exists       x\FragmG,\enc{u,\repl\alpha x q},\enc{v,\repl\beta x r})$ \\
$\quant x$ & $\forall x$     & $q \in \dom{v}$ & $r \in \dom{u}$ & $(\dis\forall       x\FragmG,\enc{u,\repl\alpha x r},\enc{v,\repl\beta x q})$ \\
$\quant x$ & $\neg\exists x$ & $q \in \dom{v}$ & $r \in \dom{u}$ & $(\dis{\neg\exists} x\FragmG,\enc{v,\repl\beta x q},\enc{u,\repl\alpha x r})$ \\
$\quant x$ & $\neg\forall x$ & $q \in \dom{u}$ & $r \in \dom{v}$ & $(\dis{\neg\forall} x\FragmG,\enc{v,\repl\beta x r},\enc{u,\repl\alpha x q})$ \\ \bottomrule
  \end{tabular}
  \caption{A single round of the $\Fragm$-game in configuration $S=(\FragmG,\enc{u,\alpha},\enc{v,\beta})$.}
  \label{tab:masterplan}
\end{table}
Note that the chosen variable $x$ can, but does not need to be from $\Varset$.

\newpage\noindent
The winning conditions for the two players are as follows:
\begin{itemize}
  \item If Spoiler cannot choose a quantifier and a variable such that the corresponding reduct is
        not empty (i.\,e.~$\FragmG$ does not contain any more quantified formulae), then Duplicator
        wins.
  \item If either player cannot choose a position because its word's domain is empty, the other 
        player wins.
  \item If the game reaches an $\Varset$-configuration $(\FragmG, \enc{u, \alpha}, \enc{v, \beta})$
        such that there is a literal (i.\,e.~an unquantified formula) $\varphi \in \FragmG$ with 
        $\fv(\varphi) \subseteq \Varset$ and $\enc{u, \alpha} \models \varphi$ but $\enc{v, \beta}
        \not\models \varphi$, then Spoiler wins.
  \item Duplicator wins, if Spoiler does not win. In particular Duplicators wins every game which
        goes on forever.
\end{itemize}
Whenever a player has won, the game stops immediately. Given the winning conditions either Spoiler
or Duplicator has a winning strategy on the words $(u, v)$ in the $\Fragm$-game, the game is 
\emph{determined}. For a configuration $S = (\Fragm, \enc{u, \alpha}, \enc{v, \beta})$ we write
$\enc{u, \alpha} \Dwins\Fragm \enc{v, \beta}$ if Duplicator has a winning strategy on S, or
$\enc{u, \alpha} \Swins\Fragm \enc{v, \beta}$ if Spoiler has one. One should note that 
$\Dwins\Fragm$ is a preorder, i.\,e.~a reflexive and transitive binary relation.

Suppose that the quantifier depth of all formulae in a fragment $\Fragm$ is bound by $n \in \Nat$.
Because we lose one quantifier level in very round of the game, the $\Fragm$-game can last at most
$n$ rounds. This means that the $\Fragm_n$-game can be seen as an $n$-round version of the 
$\Fragm$-game for \emph{any} fragment. We are going to have a closer look on this kind of games. The
following result from~\cite{HuschenbettKufleitner14stacs:short} is an adaption of the \EF
Theorem for fragments of bounded quantifier depth.

\begin{restatable}{theorem}{TheoremEFBounded}
\label{thm:EF_bounded}
Let $\Fragm$ be a fragment of bounded quantifier depth. For all words
$u,v \in \words$ the following are equivalent:%
\begin{enumerate}
\item $u \models \varphi$ implies $v \models \varphi$ for all
  sentences $\varphi \in \Fragm$ and
\item Duplicator has a winning strategy in the $\Fragm$-game on
  $(u,v)$.
\end{enumerate}
\par
\end{restatable}

\noindent
The proof for the implication ``\refenum{2} $\bm\Rightarrow$ \refenum{1}'' does not depend on the
bounded quantifier depth of the fragment. The implication thus holds also for the unbounded case.
For the implication ``\refenum{1} $\bm\Rightarrow$ \refenum{2}'' the boundedness, however, is 
substantial. An example for this is an \EF game on $(a^\zeta,a^{\zeta+\zeta})$ for an 
arbitrary $n \in \Nat$ where $\zeta$ denotes the natural order type of all integers. Here, 
Duplicator has a winning strategy for all $\FO_n$-games with $n \in \Nat$ which yields that
$a^\zeta \models \varphi$ implies $a^{\zeta + \zeta} \models \varphi$ for all sentences 
$\varphi \in \FO$. On the other hand, Spoiler has a winning strategy in the (infinite) $\FO$-game.

To gain a version of the theorem which does not limit the quantifier depth of 
the fragment, one has to restrict the words $u$ and $v$ to a special subset of words.
For this we define:
\begin{definition}
  \label{def:tau-rationality}
  Let $\tau$ be a linear ordering. A word $u \in
  \words$ is \emph{$\tau$-rational} if it can be constructed from
  the finite words in $\words$ using the operations of concatenation
  and $\tau$-power, only.
\end{definition}
\noindent Next we need to define some special oder types:
\begin{definition}
  The order types of the sets $\Nat$, $\Int$, $\Rat$ and $\Int_{\leq 0} = 
  -\Nat$ ordered naturally are denoted by $\omega$, $\zeta$, $\eta$ and 
  $\omega^*$. Let $a \in \Lambda$.   The order type of the word $a^\omega\bigl(a^\zeta\bigr){}^\eta a^{\omega^*}$ is $\theLO = \omega + \zeta\cdot\eta + \omega^*$,
  and $\sigma = \omega + \zeta + \omega^*$ is the order type of $a^\omega a^\zeta a^{\omega^*}$.
\end{definition}

With these definitions it is now possible to state the following 
theorem from \cite{HuschenbettKufleitner14stacs:short}.

\begin{restatable}{theorem}{TheoremEFUnbounded}
\label{thm:EF_unbounded}
Let $\Fragm$ be fragment. For all $\theLO$-rational words $u,v \in \words$ the 
following are equivalent:
\begin{enumerate}
\item $u \models \varphi$ implies $v \models \varphi$ for all
  sentences $\varphi \in \Fragm$ and
\item Duplicator has a winning strategy in the $\Fragm$-game on
  $(u,v)$.
\end{enumerate}
\par
\end{restatable}

\noindent
The implication ``\refenum{2} $\bm\Rightarrow$ \refenum{1}'' is a special case of the according
implication for Theorem~\ref{thm:EF_bounded} which holds -- as already mentioned -- even if one
drops the requirement of $\Fragm$ being of bounded quantifier depth. Therefore, a proof for this
direction can be obtained along the lines of a proof for the classic version of the \EF Theorem.
At the heart of the proof
for the implication \mbox{``\refenum{1} $\bm\Rightarrow$ \refenum{2}''} lies the
following lemma:
\begin{lemma}
\label{lemma:limit_points}
Let $\Fragm$ be a fragment, $x \in \Vars$ and $\enc{u,\alpha}$ an $\Varset$-valuation on a $\theLO$-rational word $u
\in \words$. For every infinite sequence $(q_i)_{i \in \Nat} \in
\dom{u}^\Nat$ there exists a position $q \in \dom{u}$ such that for
all $n \in \Nat$ there are arbitrarily large $i \in \Nat$ with
$\enc{u,\repl\alpha x{q_i}} \Dwins{\Fragm_n} \enc{u,\repl\alpha x q}$.
\end{lemma}
\noindent
The key idea here is the following: by Theorem~\ref{thm:EF_bounded} there is
a winning strategy in the $\Fragm_n$-game for every $n \in \Nat$. Obviously, the winning 
strategy for the $\Fragm_k$-game is also a winning strategy for the 
$\Fragm_n$-game, if $k \geq n$. So, for every $n \in \Nat$ there is an infinite
number of winning strategies. Lemma~\ref{lemma:limit_points} 
states that these winning strategies can be combined into a single winning strategy for all $\Fragm_n$-games. This is the main ingredient for showing how the (finite) formulae in $\Fragm$ define a winning strategy in an infinite game; details to this are analogous to the results for two-variable logic which can be found below.

\section{Restriction on Two Variables}

Let $\FO^3$ be the subset of formulae in $\FO$ which only use variables from 
$\smallset{x, y, z}$ and let $\FO^2$ be the subset using only variables form
$\smallset{x, y}$.
It is well known that for every sentence in $\FO$ there is an equivalent 
sentence in $\FO^3$. The study of $\FO^2 = \FO^2[{<}]$, therefore, arises as a natural 
limitation.
For fragments $\Fragm \subseteq \FO^2[<]$ the order type $\theLO$ 
can be simplified to $\twoVarLO = \omega + \zeta + \omega^*$. As it will later on be 
necessary to distinguish the first $n$ and the last $n$ positions of 
$a^\sigma$ from the rest, we will say that the union of those positions forms 
the \emph{$n$-border}. Before we can go into detail there is a need for yet 
another definition. Let $p_1$ and $p_2$ be two elements of the same linear
ordering $(P, \leq_P)$; define
$$
  \ord(p_1, p_2) = \begin{cases}
    < & \text{if } p_1 <_P p_2\\
	= & \text{if } p_1 = p_2\\
	> & \text{otherwise, i.\,e.~} p_1 >_P p_2\text{.}
  \end{cases}
$$

The simplification from $\theLO$ to $\twoVarLO$ manifests in the following
variation of Theorem~\ref{thm:EF_unbounded}:
\begin{theorem}\label{thm:EF_FO2}
  Let $\Fragm \subseteq \FO^2[<]$ be a fragment.
  For all \twoVarLOText-rational words $u, v \in \words$ the following are 
  equivalent:
  \begin{enumerate}
  \item $u \models \varphi$ implies $v \models \varphi$ for all
    sentences $\varphi \in \Fragm$ and
  \item Duplicator has a winning strategy in the $\Fragm$-game on
    $(u,v)$.
  \end{enumerate}
\par
\end{theorem}
\noindent
While the proof for the implication ``\refenum{2} $\bm\Rightarrow$ \refenum{1}''
is analogous to Theorem~\ref{thm:EF_unbounded}, the proof for the implication
``\refenum{1} $\bm\Rightarrow$ \refenum{2}'' differs slightly. It heavily relies
on the following lemma:
\begin{lemma}\label{lemma:nborder}
  Let $\Fragm \subseteq \FO^2[<]$ be a fragment 
  and let $\enc{u, \alpha}$ and $\enc{u, \beta}$ be $\Varset$-valuations on $u = v^\twoVarLO$ for a non-empty \twoVarLOText-rational word $v \in 
  \words$ and $\Varset \subseteq \smallset{x, y}$. For $\gamma \in 
  \smallset{\alpha, \beta}$ and $z \in \Varset$ let $\gamma(z) = (s_z^\gamma, 
  p_z^\gamma)$, and define $\gamma_v: \Varset \to \dom{v}$ via $\gamma_v(z) = 
  s_z^\gamma$.
  Suppose for $n \in \mathbb{N}$ the following three conditions hold:
  \begin{enumerate}
    \item $\enc{v, \alpha_v} \Dwins\Fragm_n \enc{v, \beta_v}$.
    \item If $p_z^\gamma$ is in the $n$-border for $\gamma \in \smallset{ \alpha, \beta}$ and $z \in 
      \Varset$, then $p_z^\alpha = p_z^\beta$.
    \item $\ord(p_x^\alpha, p_y^\alpha) = \ord(p_x^\beta, p_y^\beta)$ if 
      $\Varset = \smallset{x, y}$.
  \end{enumerate}
  Then we have $\enc{u, \alpha} \Dwins\Fragm_n \enc{u, \beta}$.
\end{lemma}
\begin{proof}
  The proof is given by induction on $n$.
  
  \resetcases\basecase{$n = 0$} As $\dom{u}$ is infinite, $\enc{u, \alpha} 
  \Dwins\Fragm_0 \enc{u, \beta}$ holds if $\enc{u, \alpha} \models \varphi$ 
  implies $\enc{u, \beta} \models \varphi$ for any literal $\varphi \in \Fragm$.
  One should note that, obviously, $\enc{u, \alpha}$ can only be a model for a 
  literal, if $\alpha$ -- and thereby also $\beta$ -- is defined on the 
  variables appearing in the literal. 
  If $\enc{u, \alpha} \models x = y$ then $s_x^\alpha = s_y^\alpha$ and 
  $p_x^\alpha = p_y^\alpha$. This implies $s_x^\beta = s_y^\beta$ by \refenum{1}
  and $p_x^\beta = p_y^\beta$ by \refenum{3} For $x < y$ and $x \leq y$ one
  has to distinguish the cases $p_x^\alpha = p_y^\alpha$ and $p_x^\alpha < 
  p_y^\alpha$, but in both cases it is easy to verify that $\enc{u, \beta}
  \models x < y$, or $x \leq y$ respectively, holds. $\enc{u, \beta} \models
  \lambda(x) = a$ follows for $\enc{u, \alpha} \models \lambda(x) = a$ directly
  via \refenum{1} and $u(\alpha(x)) = v(s_x^\alpha)$. 
  
  \inductivestep{$n > 0$} Without loss of generality, let Spoiler choose 
  variable $x$, quest $q = (s, p) \in \dom{u} = \dom{v} \times \left( \Nat 
  \uplus \Int \uplus -\Nat \right)$ and quantifier $Q \in \smallset{
  \exists, \forall, \neg\exists, \neg\forall}$. The proof for $Q \neq \exists$ 
  is either identical or symmetric to the one for 
  $Q = \exists$, as one can easily verify. Therefore, let $Q = \exists$. If $y 
  \not\in \Varset$, Spoiler's answer can be equal to $q$ and we are done. Thus,
  it is safe to assume $y \in \Varset$. By
  \refenum{1}~there has to be $s' \in \dom{v}$ such that $\enc{v, 
  \repl{\alpha_v}{x}{s}} \Dwins{(\dis\exists x\Fragm)_{n - 1}} \enc{v, 
  \repl{\beta_v}{x}{s'}}$ holds. Spoiler's answer $(s', r)$ depends
  on the value of $p_y^\beta$. In all cases, $r$ should be chosen in such a way
  that $\ord(p, p_y^\alpha) = \ord(r, p_y^\beta)$ is maintained and that $r$ 
  being in the $(n-1)$-border implies $r = p$. If these conditions are 
  satisfied, induction can be applied on $\repl{\alpha}{x}{q}$ and 
  $\repl{\beta}{x}{(s', r)}$ yielding $\enc{u, \repl{\alpha}{x}{q}} 
  \Dwins{(\dis\exists x\Fragm)_{n-1}} \enc{u, \repl{\beta}{y}{(s', r)}}$ which
  concludes the proof.
  
  If $p_y^\beta$ is in the $n$-border choose $r := p$. This maintains the 
  ordering, as $\ord(p, p_y^\alpha) = \ord(r, p_y^\beta)$ by \refenum{2} 
  If $p_y^\beta$ is not in the $n$-border, but $p$ is in the $(n - 1)$-border,
  the choice of $r := p$ trivially maintains the ordering, as well. Finally, if 
  neither $p_y^\beta$ is in the $n$-border nor $p$ in the $(n - 1)$-border,
  choose $r$ as the predecessor or successor of $p_y^\beta$ or equal to 
  $p_y^\beta$ depending of the ordering of $p$ and $p_y^\alpha$. By this choice 
  $r$ cannot be in the $(n - 1)$-border.
\end{proof}

Next, we restate a proposition of \cite{HuschenbettKufleitner14stacs:short} in
a simplified version. For the sake of brevity, its proof is omitted.
\begin{restatable}{proposition}{PropositionConcatenation}
\label{prop:concatenation}
Let $\Fragm \subseteq \FO^2[<]$ be a fragment, and let $\Varset_1, \Varset_2 \subseteq 
\smallset{x, y}$ be two mutually disjoint sets of variables. Furthermore, for
$i = 1, 2$ let
$\enc{u_i,\alpha_i},\enc{v_i,\beta_i}$ be $\Varset_i$-valuations with
$\enc{u_i,\alpha_i} \Dwins\Fragm \enc{v_i,\beta_i}$. 
Then $\enc{u_1 u_2,\alpha_1 \cup \alpha_2} \Dwins\Fragm \enc{v_1 v_2,\beta_1 
\cup \beta_2}$ holds.  \qed
\end{restatable}

This allows for the proof of a lemma which is in many aspects analogous to
Lemma~\ref{lemma:limit_points}.
\begin{lemma}\label{lemma:FO2_limit_points}
  Let $\Fragm \subseteq \FO^2[<]$ be a fragment
  and let $\enc{u, \alpha}$ be an $\Varset$-valuation on a 
  \twoVarLOText-rational word $u \in \words$ with $\Varset \subseteq \smallset{
  x, y}$.
  For every infinite sequence $(q_i)_{i \in \Nat} \in \dom{u}^\Nat$ there is a
  position $q \in \dom{u}$ such that for every $n \in \Nat$ there is an infinite
  set $I \subseteq \Nat$ with $\enc{u, \repl{\alpha}{x}{q_i}} \Dwins{\Fragm_n}
  \enc{u, \repl{\alpha}{x}{q}}$ for every $i \in I$.
\end{lemma}
\begin{proof}
  A position $q$ with the above property is called a \emph{$\enc{u, \alpha}$-limit point} of the 
  sequence $(q_i)_{i \in \Nat}$. To proof the lemma we have to show the existence of such a $\enc{u,
  \alpha}$-limit point for any sequence $(q_i)_{i \in \Nat} \in \dom{u}^\Nat$. Without loss of generality
  we can assume that $x \not\in \Varset$, because neither $\repl\alpha x {q_i}$ nor $\repl\alpha x q$
  depends on the value of $\alpha(x)$.
  
  The rest of the proof is given by induction on the \twoVarLOText-rational 
  construction of $u$.

  \resetcases \basecase{$u$ is finite} Trivially, we have that $\dom{u}$ is also finite, and that there is
  a $q \in \dom{u}$ such that $q = q_i$ for infinitely many $i \in \Nat$. This $q$ is a 
  $\enc{u,\alpha}$-limit point of $(q_i)_{i \in \Nat}$.

  \lpara{Inductive step 1}{$u = v_1 v_2$ with $\twoVarLO$-rational words $v_1,v_2$}
  We can split the valuation $\enc{u, \alpha}$ into two valuations $\enc{v_1, \beta_1}$ and 
  $\enc{v_2, \beta_2}$ with $\alpha = \beta_1 \cup \beta_2$. For $\ell = 1$ or $\ell = 2$ there is a set
  $I \subseteq \Nat$ such that $q_i \in \dom{v_\ell}$ for all $i \in I$. By induction, the subsequence $(q_i)_{i \in
  I}$ has a $\enc{v_\ell, \beta_\ell}$-limit point, and, by Proposition~\ref{prop:concatenation} this limit
  point is also a $\enc{u, \alpha}$-limit point of $(q_i)_{i \in \Nat}$.
  
  \lpara{Inductive step 2}{$u = v^\twoVarLO$ for a \twoVarLOText-rational word $v$} As $\dom{u} = \dom{v} \times 
  \left( \Nat \uplus \Int \uplus -\Nat \right)$, let $q_i =: (s_i, p_i)$ for
  $i \in \Nat$. Define $\alpha_v: \Varset \to \dom{v}$ as the restriction of 
  $\alpha$ to values in $\dom{v}$ by discarding the ($\Nat \uplus \Int \uplus 
  -\Nat$)-part. In each of the following cases there is a definition of a value
  $p$ and an infinite index set $J$ such that induction yields a $\enc{v, 
  \alpha_v}$-limit point $s$ of the sequence $(s_i)_{i \in J}$. Define $I 
  \subseteq J$ as the infinite set of indices $i$ with $\enc{v, 
  \repl{\alpha_v}{x}{s_i}} \Dwins{\Fragm_n} \enc{v, \repl{\alpha_v}{x}{s}}$ and 
  define $q$ as $(s, p)$. One should observe that in all cases $p$ is chosen in 
  a way that satisfies the conditions of Lemma~\ref{lemma:nborder}. Thus, 
  $\enc{u, \repl{\alpha}{x}{(s_i, p_i)}} \Dwins{\Fragm_n} \enc{u, 
  \repl{\alpha}{x}{(s, p)}}$ follows for all $i \in I$.
  
  \resetcases\case{There is $p$ such that $p = p_i$ for infinitely many $i$}
  Let $J$ be the infinite set of these $i$.
  
  \case{There is an unbounded subsequence of $(p_i)_{i \in \Nat}$ with values in
  $\Nat$ or one with values in $-\Nat$} The argumentation for values in $-\Nat$
  is dual to the one with values in $\Nat$. Therefore, demonstration is only 
  given for the $\Nat$-case. If $\alpha(y) = (s_y, p_y)$ is defined and if $p_y$
  is in the $\Nat$-part, let $m$ be the maximum of $n$ and $p_y$, and let 
  $p = 0$ in the $\Int$-part. If $p_y$ is in the $\Int$-part, let $m = n$ and $p = p_y - 1$. 
  Otherwise, let $m = n$ and $p = 0$ in the $\Int$-part. Let $J$ be the infinite set of those 
  indices $i$ of the unbounded subsequence with $p_i > m$. Special note shall 
  be given to the fact that, by choice, the ordering of $p_i$ and $p_y$ is equal 
  to the one of $p$ and $p_y$ for all $i \in J$.
  
  \case{There is an unbounded subsequence of $(p_i)_{i \in \Nat}$ with values
  in $\Int$} By duality one can assume that the subsequence is upwards 
  unbounded. If $\alpha(y) = (s_y, p_y)$ is defined and $p_y$ is in the 
  $\Int$-part, let $p = p_y + 1$. Otherwise, choose $p$ as the $0$ of the 
  $\Int$-part. Finally, choose $J$ as the infinite set of all indices $i$ in the
  upwards unbounded subsequence with $p_i \geq p$. Again, this choice ensures
  that the ordering of $p_i$ and $p_y$ is equal to the one of $p$ and $p_y$ for
  all $i \in J$.
\end{proof}

One of the major differences between Lemma~\ref{lemma:FO2_limit_points} and Lemma~\ref{lemma:limit_points} is that in the latter case one can take a more local view since if we factorize $\varrho = \varrho_1 + 1 + \varrho_2$ with   infinite factors $\varrho_1$ and $\varrho_2$, then both $\varrho_1$ and $\varrho_2$ are isomorphic to $\varrho$. We cannot apply this simplification here in Lemma~\ref{lemma:FO2_limit_points}.

With this lemma in place, proving the implication ``\refenum{1} 
$\bm\Rightarrow$ \refenum{2}'' of Theorem~\ref{thm:EF_FO2} is a simple matter of 
adapting the corresponding parts of the proof for Theorem~\ref{thm:EF_unbounded}.
\begin{proof}[Proof of Theorem~\ref{thm:EF_FO2}, ``\refenum{1} $\bm\Rightarrow$ \refenum{2}''.]
  We call a configuration $(\FragmG,\enc{u,\alpha},\enc{v,\beta})$ of the 
  $\Fragm$-game on $(u,v)$ \emph{good} for Duplicator if
  $\enc{u,\alpha} \Dwins{\FragmG_n} \enc{v,\beta}$ for every $n \in
  \Nat$. Good configurations especially satisfy $\enc{u,\alpha} \Dwins{\FragmG_0}
  \enc{v,\beta}$, so Spoiler does not win immediately in good configurations.
  Thus, staying in good configurations forms a winning strategy for Duplicator.
  By Theorem~\ref{thm:EF_bounded} the initial configuration 
  $(\Fragm,u,v)$ is good. 

  Now suppose we are in the good configuration $(\FragmG, \enc{u, \alpha}, \enc{v, \beta})$ and 
  Spoiler chooses quantifier $\quant$, variable $x$ and quest $q$. Now, Duplicator has to choose 
  such an answer that the follow-up configuration is good again. This choice, obviously, depends on 
  $\quant$; here, we only demonstrate the case for $\quant = \exists$, as the other cases are 
  similar. For $\quant = \exists$ Spoiler's quest $q$ has to be in $\dom{u}$. Because we are in a 
  good configuration, we have $\enc{u, \alpha} \Dwins{\FragmG_{n+1}} \enc{v, \beta}$ for every $n
  \in \Nat$. This means that for every $n \in \Nat$ Duplicator has an answer $r_n \in \dom{v}$ such 
  that $\enc{u, \repl\alpha x q} \Dwins{\dis\exists x{\FragmG_{n+1}}} \enc{v, \repl\beta x{r_n}}$ holds.
  Because of $\dis\exists x{\FragmG_{n + 1}} = (\dis\exists x\FragmG)_n$ this is equivalent to
  $\enc{u, \repl\alpha x q} \Dwins{(\dis\exists x\FragmG)_n} \enc{v, \repl\beta x{r_n}}$. 
  Lemma~\ref{lemma:FO2_limit_points} yields an $r \in \dom{v}$ such that for every $n \in \Nat$
  there is an infinite set $I \subseteq \Nat$ with $\enc{v, \repl\alpha x{r_i}} \Dwins{(\dis\exists x
  \FragmG)_n} \enc{v, \repl\alpha x r}$ for all $i \in I$. Choosing this $r$ as Duplicator's answer 
  results in a good configuration: for every $n \in \Nat$ there is an $i \geq n$ with $\enc{v, \repl\beta 
  x{r_i}} \Dwins{(\dis\exists x\FragmG)_n} \enc{v,\repl\beta x r}$ and $\enc{u, 
  \repl\alpha x q} \Dwins{(\dis\exists x\FragmG)_i} \enc{v,\repl\beta x{r_i}}$ holds by choice of $r_i$.
  This implies $\enc{u, \repl\alpha x q} \Dwins{(\dis\exists x\FragmG)_n} \enc{v,\repl\beta x{r_i}}$ 
  and, in turn, $\enc{u, \repl\alpha x q} \Dwins{(\dis\exists x\FragmG)_n} \enc{v, \repl\beta x r}$.
\end{proof}

\subparagraph{Above $\FO^2$ and below $\sigma$.}
While it is sufficient to consider \twoVarLOText-rational words for fragments
$\Fragm \subseteq \FO^2[<]$, this restriction in 
consideration is insufficient for $\FO^3$ and, hence, for $\FO$. To see this, 
consider $u = a^\sigma$ and $v = a^\sigma a^\sigma = a^{\omega + \zeta + \zeta 
+ \omega^*}$ for an $a \in \Lambda$. These two words cannot be distinguished by 
any formula in $\FO$, 
since there is a winning strategy for Duplicator in the $\FO_n$-game for every 
$n \in \Nat$. Therefore, they especially cannot be distinguished by any formula 
in $\FO^3$. On the other hand, there is a winning strategy for Spoiler in the 
infinite $\FO^3$-game on $(u, v)$: In the first round, Spoiler can choose 
variable $x$, quantifier $\forall$ and a quest in the first $\Int$-part of 
$\dom{v}$. Duplicator's answer has to be in the $\Int$-part of $\dom{u}$, 
since otherwise there would be only finitely many positions left on one side 
which would allow Spoiler to win in finitely many rounds. In the second round,
Spoiler can choose variable $y$, quantifier $\forall$ and a quest in the second
$\Int$-part of $\dom{v}$. Again, Duplicator's answer needs to be in the 
$\Int$-part of $\dom{u}$ for the same reason as before. But: this leaves only
finitely many positions between Duplicator's first and second answer in 
$\dom{u}$, while there still are infinitely many positions between Spoiler's two
quests. Therefore, Spoiler can win in finitely many rounds.

On the other hand, one might assume that it is possible to further restrict 
$\twoVarLO$ to $\twoVarLO' = \omega + \omega^*$, the order type of $a^\omega
a^{\omega^*}$ for $a \in \Lambda$, in the case of fragments $\Fragm \subseteq 
\FO^2[<]$. But this is not the case, as can be 
seen when considering the words $u = a^{\twoVarLO'}$ and $v = 
a^{\twoVarLO'}a^{\twoVarLO'} = a^{\omega + \zeta + \omega^*}$. In $\dom{v}$
there are positions which are infinitely far away from both ends of the word,
namely the positions belonging to the $\Int$-part, but there are no such
positions in $\dom{u}$. This constitutes a winning strategy for Spoiler in the
$\FO^2$-game on $(u, v)$.

\enlargethispage{2\baselineskip}

\subparagraph{Applications.}

Let $\Fragm \subseteq \FO^2[<]$ be a fragment and let $s,t \in
\terms$. One of the main applications of Theorem~\ref{thm:EF_FO2} is the equivalence of the following two properties:
\begin{enumerate}
\item The identity $s=t$ holds in the syntactic monoid of every
  language definable in $\Fragm$.
\item Duplicator has winning strategies in the $\Fragm$-games on 
  $(\evalt s\sigma,\evalt t\sigma)$ and $(\evalt t\sigma,\evalt
  s\sigma)$.
\end{enumerate}
Here, $\evalt s\sigma$ is the generalized word obtained by replacing $\pi$ by $\sigma$. The proof is identical to the corresponding statement in~\cite{HuschenbettKufleitner14stacs:short}; a definition of the syntactic monoid can also be found in this paper.

{\small
\newcommand{\Ju}{Ju}\newcommand{\Ph}{Ph}\newcommand{\Th}{Th}\newcommand{\Ch}{Ch}\newcommand{\Yu}{Yu}\newcommand{\Zh}{Zh}\newcommand{\St}{St}\newcommand{\curlybraces}[1]{\{#1\}}

}

\end{document}